\documentclass[12pt,a4paper]{article}

\usepackage[utf8]{inputenc}
\usepackage{times}

\usepackage{amsthm,amsfonts,amssymb,amsmath} 

\newcommand{\be}{\begin{equation}}
\newcommand{\ee}{\end{equation}}

\newcommand{\ba}{\begin{eqnarray}}
\newcommand{\ea}{\end{eqnarray}}

\usepackage{cite}
\usepackage{color}
\newtheorem{thm}{Theorem}[section]

\newtheorem{prop}[thm]{Proposition}

\newcommand{\n}{\label}
\newcommand{\p}{\partial}

\begin{document}
\centerline{\bf  A lower bound of the Trautman-Bondi energy}

\bigskip

-\centerline{J. Tafel}

\noindent
\centerline{Institute of Theoretical Physics, University of Warsaw,}
\centerline{Ho\.za 69, 00-681 Warsaw, Poland, email: tafel@fuw.edu.pl}

\bigskip

\begin{abstract}
 We obtain an energy inequality on null surfaces $u=const$ in the Bondi-Sachs formalism. We show that for a sufficiently regular event horizon $H$  there is an affine radial coordinate which is constant on $H$. Then the energy inequality can be  prolongated to the horizon giving an estimation which is closely related to the Penrose inequality. We test it for the Kerr solution written in the Fletcher-Lun coordinates.
\end{abstract}

\bigskip

\noindent
 Keywords:
the Trautman-Bondi energy,   the Penrose inequality, the Kerr metric

\null

\noindent
PACS numbers: 04.20.Ha

\null

\section{Introduction}
Since the positive energy theorem was proved by Schoen and Yau \cite{sy1,sy2}  there have been many attempts to strengthen this result in the case of black holes (see \cite{m,d} for a review). They were motivated by the weak cosmic censorship conjecture of Penrose  which led him to  the conclusion \cite{p} that
\be
M_{ADM}\geq \sqrt{\frac{A}{16\pi}}\ ,\n{1b}
\ee
where $M_{ADM}$ is the Arnowitt-Deser-Misner mass and $A$ is an area of a black hole horizon. It seems that a knowledge of the global structure of spacetime is not necessary to obtain inequality (\ref{1b}). It should be also satisfied for asymptotically flat initial data admitting a trapped surface with area $A$.  In order to prove it several   quasi local definitions of mass \cite{g,haw} were used. In the case of spherical symmetry  inequality (\ref{1b}) was obtained for maximal data by Malec and Murchadha \cite{mm}  and then generalized by Hayward \cite{h}. One of the most important recent results is a proof of the Penrose inequality  for  time-symmetric data  \cite{hi,b}. The general case remains an open problem, although there are some indications of how to solve it  \cite{bk}. 

Soon after publishing the paper of Schoen and Yau the positivity of the Bondi mass (also called the Trautman-Bondi energy) was  established \cite{in,lv}. 
There were  also  attempts to prove the Penrose inequality for the characteristic data  \cite{lv1,ber} (see \cite{m} for a discussion). 

Recently, an alternative  proof of the positivity of energy of a null surface $u=const$  was  presented  by Chru\'{s}ciel and Paetz \cite{cp} and Tafel \cite{t} (see also \cite{kt}).  A drawback of this proof is that it assumes a regular cone properties of the null surface.  In the present paper we do not make  this assumption. We express the Trautman-Bondi energy in terms of null expansions of 2-dimensional cuts of $u=const$. Then we consider a limit of integral consequences of the Einstein equations   when the event horizon is approached. Under some regularity assumptions, the positivity of the Trautman-Bondi energy follows. Moreover, we obtain an inequality which becomes the Penrose inequality if asymptotic data on  the horizon and at null infinity are appropriately correlated.   
The latter  result is not a proof of the Penrose inequality in terms of initial data. It  rather shows that this inequality can be satisfied, even if the  gravitational collapse did not approach  the equilibrium state represented by  the Kerr metric.

In section 2 we introduce the Bondi-Sachs approach to asymptotical flatness \cite{bbm,s} and its version \cite{nu}  with an affine coordinate  instead of the luminosity distance. We present several forms of the Trautman-Bondi energy. Following \cite{t}, in section 3, we use one of the Einstein equations in order to obtain a lower bound of the energy as an integral over a  cut of the surface $u=const$. In section 4, we assume that  metric admits an event horizon $H$ with a reasonable level of regularity. We show that there exists an affine distance $r$ which is  constant on $H$. This allows to continue the energy inequality of section 3 to the horizon. In this way we obtain a lower bound of the energy (see equation (\ref{34})) given by value of $r$ on $H$ with a correction depending on data  at null infinity. Under an assumption about a long time behavior of asymptotic data   the Penrose inequality follows. In section 5, we test the above results for the  Kerr 
metric written in terms of
coordinates closely related to those of 
Fletcher and Lun \cite{fl}. In this case, surfaces $u=const$ have conical singularities at $\theta=0,\pi$. A direct calculation made with  Mathematica 10 shows that  energy inequality (\ref{34})  is   stronger than the Penrose inequality.

\section{The Trautman-Bondi energy}
The approach to gravitational radiation originated by Bondi assumes a foliation of spacetime by null surfaces $u=const$ of the form $R\times S_2$. The factor $R$ corresponds to null geodesics generated by  $u^{,\alpha}\p_{\alpha}$ and $S_2$ is the 2-dimensional sphere.  In adapted coordinates (called the generalized Bondi-Sachs coordinates in \cite{fl}) metric  can be written in the form 
\begin{equation}
g=du(g_{00}du+2g_{01}dr+2\omega_{A} dx^A)+g_{AB}dx^Adx^B\ .\label{1a}
\end{equation}
  Coordinate $x^0=u$ is interpreted as a retarded time  and $x^A$, where $A=2,3$, are coordinates on $S_2$. The standard metric of the sphere is given by $s_{AB}$. One can still impose a condition on the coordinate $x^1=r$. That of Bondi and Sachs reads 
  \be
\sigma=r^2\sigma_s\ ,\n{3c}
\ee
where
  \be
\sigma=\sqrt{\det  g_{AB}}\ ,\ \ \sigma_s=\sqrt{\det s_{AB}}\ .\n{3b}
\ee
In this case  $r$ is called the luminosity distance. For an asymptotically flat metric condition (\ref{3c}) allows to write conveniently  the total energy corresponding to the surface $u=const$  
\be
E(u)=\frac{1}{4\pi}\int_{S_2} M_B \sigma_sd^2x\ .\label{5}
\ee
Here $M_B$ is a function of $u$ and $x^A$ given by the asymptotic expansion $g_{00}\simeq 1-\frac{2M_B}{r}$. Function 
$E$ is called the Bondi mass or the Trautman-Bondi energy since  it coincides with  energy defined in a different way by Trautman \cite{tr}. If the Ricci tensor vanishes sufficiently quickly  at null infinity	
then  $E(u)$  diminishes with time what is interpreted as an effect of the outgoing gravitational radiation \cite{bbm,s,tr}.

In order to find  inequalities satisfied by  $E$, it is more convenient to replace condition (\ref{3c}) by the Newman-Unti condition $g_{01}=1$ (see \cite{bl} for a recent discussion of this gauge). Then metric reads
\begin{equation}
g=du(g_{00}du+2dr+2\omega_{A} dx^A)+g_{AB}dx^Adx^B\label{1}
\end{equation}
  and $r$ is the affine parameter along null geodesics tangent to $u=const$. If not  stated otherwise, in the rest of the paper  we  assume that metric is given  by (\ref{1}) in a part of spacetime of the form $\mathcal{M}'=(u_0,u_1)\times (r_0,\infty)\times S_2$. Function $g_{00}$, form $\omega_Adx^A$ and tensor $g_{AB}dx^Adx^B$ are smooth (with exceptions in section 5) on $\mathcal{M'}$.  Moreover, the following  asymptotical flatness conditions  (see \cite{t}) are assumed
\begin{equation}
g_{00}=1+\frac 12 n_{,u}-2Mr^{-1}+O(r^{-2})\label{2}
\end{equation}
\begin{equation}
 g_{AB}=-s_{AB}r^2+ n_{AB}r+m_{AB}+O(r^{-1})\label{4}
\end{equation}
\begin{equation}
\omega_{A}=\psi_A+\kappa _A r^{-1}+O(r^{-2})\ .\label{3}
\end{equation}
Here,  all coefficients are independent of $r$ and 
\be
n=-s^{AB}n_{AB}\ .\n{3a}
\ee

Given (\ref{1})-(\ref{3}), one can define the Bondi luminosity distance (now denoted by $r_B$) as
\be
r_B=\sqrt{\frac{\sigma}{\sigma_s}}\ .\label{4c}
\ee
It follows from (\ref{4}) that asymptotically 
\be
r_B= r+\frac {1}{4}n+O(r^{-1})\ .\label{4b}
\ee
In order to define  energy (\ref{5}), we can use  expressions for $M_B$  following from (20) and (32) in \cite{t} 
\be
M_B=\frac {1}{4\sigma_s}\lim_{r\rightarrow\infty}{(\sigma_{,r}+g^{11}\sigma_{,r}+2\sigma_{,u})}\label{6a}
\ee
\be
M_B=\frac {1}{4\sigma_s}\lim_{r\rightarrow\infty}{(2\sigma_sr+g^{11}\sigma_{,r}+2\sigma_{,u})}+\frac 18n\ ,\label{6}
\ee
where
\be
g^{11}=\omega^A\omega_A-g_{00}\ .\n{8}
\ee
(Remark: There should be a minus sign  before $g_{00}$ in (21) in \cite{t}. This change does not have an impact on further considerations in \cite{t}).
 
 Expression (\ref{5}) can be written in terms of  expansions $\theta_{\pm}$ of null rays emitted orthogonally from  2-dimensional surfaces $S(u,r)$ given by constant $u$ and $r$. Outgoing rays are tangent to $\p_r$ and ingoing ones are tangent  to 
 \be
 k=g^{11}\p_{r}+2\p_{u}-2\omega^A\p_A\ .\label{9}
\ee
The scalar product of these  vectors  yields
\be
g(k,\p_r)=2\ .\n{12}
\ee
The Lie derivative of $g$ with respect to $\p_r$ and $k$ defines  tensors on $S(u,r)$ with traces given by
\be
\theta_+=(\ln{\sigma})_{,r}\ ,\ \ \theta_-=g^{11}(\ln{\sigma})_{,r}+2(\ln{\sigma})_{,u}-2\omega^A_{\ |A}\ ,\n{10}
\ee
where ${}_{|A}$ denotes the covariant derivative with respect to metric $g_{AB}$.
Hence, equation (\ref{6a}) is equivalent to 
\be
M_B=\frac {1}{4\sigma_s}\lim_{r\rightarrow\infty}{\sigma(\theta_++\theta_-+2\omega^A_{\ |A})}\label{11}
\ee
and (\ref{5}) takes the form
\be
E(u)=\frac{1}{16\pi}\lim_{r\rightarrow\infty}{\int_{S(u,r)} (\theta_++\theta_-) \sigma d^2x}\ .\label{11a}
\ee

Note that given a spherical surface $S$,  expansions $\theta_{\pm}$ are not uniquely defined  locally since null vectors orthogonal to $S$ are given up to a conformal factor. 
Thus, $(\theta_++\theta_-)$ cannot be considered as a quasi local mass related to   $S$. In order to distinguish expansions giving rise to (\ref{11a}) one should follow outgoing  null geodesics emitted orthogonally from $S$  and establish an affine parameter $r$  by means of the condition $g_{AB}\approx-s_{AB}r^2$.  Once this is done, vector (\ref{9}) is also uniquely defined as a null vector orthogonal to $S$ and satisfying (\ref{12}).

\section{Energy inequalities}
In \cite{cp,t} it was shown that if a null surface $u=const$  can be continued to the past up to a regular vortex, the corresponding Trautman-Bondi energy is nonnegative provided that some of the Einstein equations and the dominant energy condition are satisfied.  These equations are equivalent to the Raychaudhuri equation $R_{11}=T_{11}$  and equation  $g^{AB}R_{AB}=g^{AB}T_{AB}-T$, where $T=T^{\mu}_{\ \mu}$ and $R_{\mu\nu}$ and $T_{\mu\nu}$ are, respectively, the Ricci tensor and the energy-momentum tensor. They read
\be
(\ln \sigma)_{,rr}-\frac{1}{4}g_{AB,r}g^{AB}_{\ \ ,r}=-T_{11}\label{13}
\ee
\be
\sigma^{-1}(\sigma\theta_-)_{,r}-R^{(2)}+(\omega^A_{\ ,r})_{|A}+\frac{1}{2}g_{AB}\omega^A_{\ ,r}\omega^B_{\ ,r}=T-g^{AB}T_{AB}\ ,\label{14}
\ee
where $\theta_-$ is given by (\ref{10}) and 
$R^{(2)}$ is the Ricci scalar  of  metric $g_{AB}$. In this section we analyze these equations with no  assumptions regarding existence of vortices of surfaces $u=const$.

 Integrating (\ref{14})  with density $\sigma$ over a surface $S(u,r)$ of constant $u$ and $r$ yields 
\be\label{15}
\int_{S(u,r)}{(\sigma\theta_-+2\sigma_sr)_{,r}d^2x}=\int_{S(u,r)}P\sigma d^2x\ ,
\ee
where 
\be\label{16}
P=T-g^{AB}T_{AB}-\frac{1}{2}g_{AB}\omega^A_{\ ,r}\omega^B_{\ ,r}
\ee
and the Gauss-Bonet theorem 
\be
\int_{S_2}{R^{(2)}\sigma d^2x}=-8\pi\label{17}
\ee
has been used. In view of (\ref{5}) and (\ref{6}) integration of (\ref{15})  from $r$  to $\infty$ leads to
\be\label{18}
16\pi E(u)=\int_{S(u,r)}{(2\sigma_sr+\frac 12n\sigma_s+\sigma\theta_-)d^2x}+\int_{r}^{\infty}dr'\int_{S(u,r')}P\sigma d^2x\ .
\ee
Note that  the r. h. s. of (\ref{18}) can be calculated for any value of $r$.

If the dominant energy condition is satisfied, then $P\geq 0$ \cite{t} and it follows from (\ref{18}) that
\be\label{19a}
E\geq\frac 12r+\frac 18\langle n \rangle+\frac{1}{16\pi}\int_{S(u,r)}{\sigma\theta_-d^2x}\ ,
\ee
where $\langle n\rangle$ is the $u$-dependent mean value of $n$ over $S_2$
\be\n{20}
\langle n \rangle=\frac{1}{4\pi}\int_{S_2}{n\sigma_sd^2x}\ .
\ee
The integral term in (\ref{19a}) is negative for big values of $r$ because metric is asymptotically flat.
In order to show that the sum of the first two terms is positive,  let us  write equation
 (\ref{13}) in the following form 
\be
\frac{1}{4}\hat r^2\hat g_{AB,r}\hat g^{AB}_{\ \ ,r}-(\hat r^2(\ln \hat\sigma)_{,r})_{,r}=\hat r^2T_{11}\ ,\label{39}
\ee
where 
\be\n{39a}
\hat r=r+\frac 14 n
\ee
and
\be 
g_{AB}=\hat r^2\hat g_{AB}\ ,\ \ \sigma=\hat r^2\hat \sigma\ .\n{39a}
\ee
Integrating (\ref{39})  between $r$ and $\infty$ yields
\be
\hat r^2(\ln \hat\sigma)_{,r}=\int_r^{\infty}{(r'+\frac 14n)^2(T_{11}-\frac{1}{4}\hat g_{AB,r}\hat g^{AB}_{\ \ ,r}})dr'\ ,\label{40}
\ee
where approximation (\ref{4b}) was used to  obtain
\be\label{41}
\lim_{r\rightarrow\infty}\hat r^2(\ln \hat\sigma)_{,r}=0\ .
\ee
 For big values of $r$ there is $\hat r>0$. Suppose that we diminish $r$ along an integral line of $u^{,\alpha} \p_{\alpha}$  and achieve a point $p$  where $\hat r=0$. Then
\be
\hat r^2(\ln \hat\sigma)_{,r}=\hat r^2(\ln \sigma)_{,r}-2\hat r=0\ \ at\ \ \ p\ .\n{42}
\ee
If the dominant energy condition is satisfied, equation (\ref{40}) implies
\be
T_{11}=0\ ,\ \ g_{AB}=f_{AB}\hat r^2\ ,\ \ f_{AB,r}=0\n{43}
\ee
everywhere on the line from $p$ to $\infty$. It follows from (\ref{43}) that $g_{AB}$ is degenerate at $p$, so this point is either singularity or a vortex of the cone $u=const$. In any case it cannot belong to the considered domain  $\mathcal{M}'$.
Thus, without  a loss of generality we can assume that in this domain
\be
\hat r=r+\frac 14n> 0\ ,\n{44}
\ee
hence
\be
r+\frac 14\langle n\rangle> 0\ .\n{44a}
\ee

Among all (positive) affine coordinates  $\hat r=r+\frac 14n$ is distinguished by the asymptotical flatness conditions. From (\ref{4c}) and (\ref{4b}) one obtains
\be \n{53a}
\sigma=\sigma_s\hat r^2+O(1)\ \ if\ \ r\rightarrow\infty\ .
\ee
Since there is no term linear in $\hat r$ in (\ref{53a}) the coordinate
$\hat r$ can be considered as the best affine approximation  of the luminosity distance $r_B$ . If  $r=\hat r$ inequality (\ref{19a}) can be further simplified in the following way.   From (\ref{40}) one obtains
\be
\hat\sigma_{,r}\geq 0\label{45}
\ee
(note that this inequality is stronger than  $\sigma_{,r}\geq 0$ obtained by a direct integration of (\ref{13})).
Since $\lim_{r\rightarrow\infty}{\hat\sigma}=\sigma_s$, inequality (\ref{45}) implies
\be
\sigma\leq \sigma_s\hat r^2\ .\label{46}
\ee
 If $\hat r=const$, then 
integrating (\ref{46}) over coordinates $x^A$ leads to
\be
A\leq 4\pi\hat r^2\ ,\n{46a}
\ee
where $A$ is an area of $S(u,r)$. Substituting $\hat r\geq\sqrt{A/4\pi}$ into (\ref{19a}) yields
\be\label{46b}
E\geq\sqrt{\frac{A}{16\pi}}+\frac{1}{16\pi}\int_{S(u,r)}{\sigma\theta_-d^2x}\ ,
\ee

Because of the integral terms in (\ref{19a}) and (\ref{46b}) we are not able to prove that $E\geq 0$.
We will show in the next section that this situation  may change if $S(u,r)$ approaches the event horizon.

\section{Energy estimates on the event horizon}
In this section we will investigate a long $u$ limit of (\ref{19a}) for spacetime containing a black hole region. We will make the following assumptions:

(i) Spacetime $\mathcal{M}$ admits a submanifold $\mathcal{M}'$ isomorphic to $(u_0,\infty)\times (r_0,\infty)\times S_2$ with  metric (\ref{1}) satisfying conditions (\ref{2})-(\ref{3}).

(ii) $\mathcal{M}'$ has a smooth boundary $H$ (horizon) where $u=\infty$. Metric (\ref{1}) can be continued through $H$ in another system of coordinates.

(iii) There exists a  positive function $\beta$  on $\mathcal{M}'$ such that field $l_0=\beta\p_r$ extends to  $H$ as a nonvanishing null vector tangent to $H$. 

In the conformal approach \cite{p} to asymptotical flatness, the future null infinity is identified with the boundary (scri)  $J^+$ of the compactified spacetime.  Given the section $u=const$ of $J^+$, the physical metric can be always transformed to the form (\ref{1}) in a neighborhood of the section. What is important in assumption (i) is that this neighborhood does not shrink to the point $i^+$ when $u\rightarrow\infty$.  Assumption (ii) is rather standard in the theory of black holes, however, in general, the horizon does not have to be smooth. If it is smooth, we may expect that $H$ together with surfaces $u=const$ form a regular null foliation given by levels of a new coordinate $\tilde u=f(u)$  finite on $H$.  If $\tilde u$ has a nonvanishing gradient on $H$ then $l_0=\tilde u^{,\alpha}\p_{\alpha}=f_{,u}\p_r$ assures condition (iii). Thus, (iii) refers  to regularity of a passage from surfaces $u=const$ to $H$.  Note that all conditions (i)-(iii) are satisfied by  the Schwarzschild solution. 

\begin{prop}
If conditions (i)-(iii) are satisfied then there exists an affine coordinate $r$ such that
 \be
r=r_H=const\ \ on\ \ H\ .\n{27}
\ee
\end{prop}
\begin{proof}
Vector $\p_r$ is singular on $H$ since  $\p_r=u^{,\alpha}\p_{\alpha}$  and $u=\infty$ on $H$. A derivative of $u$ along a curve which  crosses surfaces $u=const$ and approaches $H$ cannot be bounded. The necessary condition to prolongate $\beta\p_r$ to  $H$ is
\be
\beta=0\ \ on\ \ H\ .\n{25}
\ee
Let  $l_0=\p_v$ in coordinates $x'^{\mu}=v,r',x'^A$ nonsingular in a neighborhood of $H$. 
From (\ref{25}) and $\p_rr=1$ it follows that $l_0(r)\rightarrow 0$ when $H$ is approached. An integral of $l_0(r)dv$ between $v_1$ and $v_2$ along  null lines tangent to  $u=const$ tends to zero. Hence, if $r$ is finite at a point $p$ of  $H$, it is constant along a null geodesic which is tangent to $H$ and passes through $p$. If we choose an affine coordinate $r$ such that $r=r_0$ on a spacelike surface $\Sigma$ intersecting $H$ then (\ref{27}) is satisfied with $r_H=r_0$.

\end{proof}
Condition (\ref{27}) says nothing about  derivatives of $r$  on $H$. If they exist then 
\be
r^{,\mu}\p_{\mu}=-\gamma l_0\ \ \ on\ \ H \ ,\n{30}
\ee
where $\gamma$ is a nonnegative function. Since 
\be 
g^{11}=r^{,\mu}r_{,\mu}=0\ \ on\ \ H
\ee
function $g^{11}/\beta$ may be nonsingular on $H$ (e.g. it happens if  $g^{11}$ and $\beta$ have first order zeros  at points of $H$). In this case
\be 
g^{11}\p_r=-\alpha l_0\ ,\n{31}
\ee
where function $\alpha\geq 0$ is nonsingular on $H$. Moreover,
since
\be
r^{,\mu}\p_{\mu}=g^{11}\p_r+\p_u-\omega^A\p_A\n{32}\ ,
\ee 
vector $\p_u-\omega^A\p_A$ is also nonsingular and null on $H$
\be
\p_u-\omega^A\p_A=(\alpha-\gamma) l_0\ \ \ on\ \ H \ .\n{33}
\ee
Note that  $\p_u-\omega^A\p_A$ is tangent to surfaces $r=const$ and orthogonal to their foliation by $u=const$. Thus, regularity of $\p_u-\omega^A\p_A$  on $H$ can  be understood as regularity of this foliation. Summarizing, if $g^{11}/\beta$ and derivatives of $r$ are nonsingular on $H$ then vectors $r^{,\mu}\p_{\mu}$, $g^{11}\p_r$ and $\p_u-\omega^A\p_A$ are extendible to the horizon and  parallel to null generators of $H$.

The main result of this paper is given by the following proposition.
\begin{prop}
 Let (i)-(iii) and  the dominant energy condition be satisfied and
 \be
0\leq \lim_{r\rightarrow\infty}{r^2R_{\mu\nu}r^{,\mu}r^{,\nu}}<\infty\ .\n{38d}
\ee
 Then
 \be\label{34}
E(u)\geq\frac 12r_H+\frac 18\langle n \rangle_{\infty}\geq 0\ ,
\ee
where $\langle n \rangle_{\infty}=\limsup_{u\rightarrow\infty}\langle n \rangle$ and condition (\ref{27}) is assumed. If, moreover,
\be\n{34a}
\lim_{u\rightarrow\infty}{(n-\langle n\rangle)}=0
\ee
then 
 \be\label{35}
E(u)\geq\sqrt{\frac {A_H}{16\pi}}\ ,
\ee
where 
$A_H$ is the  area of a spacelike spherical section of $H$.

\end{prop}
\begin{proof} 
Let $r$ be the affine coordinate satisfying (\ref{27})
and $S_H$ be a spacelike spherical surface contained in $H$. We extend $S_H$  to a spacelike section $\Sigma$ of surrounding surfaces $u=const$ and modify $r$ in such  way that $r=r_H$ on $\Sigma$ and $H$. 
 Let  $S_u$ denote an intersection of $\Sigma$ with $u=const$. It follows from (\ref{9}) and (\ref{12})  that $k_0=k/\beta$  is the unique null vector which is orthogonal to surfaces $S_u$ and satisfies $g(k_0,l_0)=2$. Thus, $k_0$ is also defined on $S_H$. It is  possible only if vector $k$  vanishes when $H$ is approached 
\be
k=0\ \ on\ \ S_H\ .\n{36}
\ee

Let us consider inequality (\ref{19a}) for $S_u$ in the form
\be\label{19c}
E(u)\geq\frac 12r+\frac 18\langle n \rangle+\frac{1}{16\pi}\int_{S_u}{L_k\chi}\ ,
\ee
where $\chi=\sigma dx^2\wedge dx^3$  and $L_k$ denotes the Lie derivative along vector (\ref{9}).
 If $u\rightarrow \infty$  inequality (\ref{19c}) yields
\be\label{37}
E_{\infty}\geq\frac 12r_H+\frac 18\langle n \rangle_{\infty}+\frac{1}{16\pi}\lim_{u\rightarrow\infty}{\int_{S_u}{L_k\chi}}\ ,
\ee
where $E_{\infty}$ is the limit of function $E(u)$ (note that  $E_{,u}\leq 0$  under condition (\ref{38d}) \cite{t1}). In virtue of (\ref{36}) the last term in (\ref{37}) vanishes (to show this it is convenient to extend $k_0$ and $\chi$ from $\Sigma$ to its neighborhood and introduce new coordinates such  that $k_0=\p_\tau$ and $\chi$ does not contain differentials $d\tau$). Hence,  inequality  (\ref{37}) reduces to 
\be\label{38}
E_{\infty}\geq\frac 12r_H+\frac 18\langle n \rangle_{\infty}\ .
\ee
 Using again $E_{,u}\leq 0$ and (\ref{44a})  implies (\ref{34}).

Integrating (\ref{46}) over $S_u$ yields
\be
A_u\leq 4\pi (r+\frac 14\langle n \rangle)^2+\frac {\pi}{4}(\langle n^2\rangle-\langle n \rangle^2)\ ,\label{47}
\ee
where $A_u$ is the area of $S_u$ and
\be
\langle n^2\rangle=\frac{1}{4\pi}\int_{S_2}{\sigma_sn^2}\n{48}
\ee
is the mean value of $n^2$ over the sphere. If condition (\ref{34a}) is satisfied and
 $u$ is big enough, 
then the last term in (\ref{47}) becomes smaller then $A_u$. In view of (\ref{44a}) one obtains
\be
\sqrt{\frac {4}{\pi}A_u-\langle n^2\rangle+\langle n \rangle^2}\leq 4r+\langle n \rangle\label{50}
\ee
and
\be
\sqrt{\frac {A_H}{4\pi}}\leq r_H+\frac 14\limsup_{u\rightarrow\infty}{\langle n \rangle}\label{51}
\ee
as a limit of (\ref{50}) when $u\rightarrow\infty$.
 A combination of inequalities  (\ref{34}) and (\ref{51})  yields (\ref{35}).

\end{proof}
Equation (\ref{34a}) means that for big values of $u$ function $n$ tends to a function $f(u)$.  If $f$ has a limit for $u\rightarrow\infty$ then 
$n\rightarrow const$ and   
\be\n{52a}
\hat r=r+\frac 14n=const\ \ on\ \ H\ .
\ee
It is difficult to predict, if  coordinates with properties (\ref{27}) and (\ref{34a})  exist.  Perhaps one can  impose (\ref{34a}) or even (\ref{52a}) using the Bondi-Metzner-Sachs (BMS) group. A supertranslation 
\be
u'=u+\alpha(x^A)+O(r^{-1})\ ,\n{52}
\ee
where $\alpha$ is a function on the sphere, 
implies  
\be\n{53}
\hat r'=\hat r-\frac 12\Delta_s\alpha+O(r^{-1})
\ee
(see e.g. \cite{t1}), where $\Delta_s$ is the standard Laplace operator on the sphere. We are not able to extract a new function $n'$ from (\ref{53}) since  the field $u^{,\alpha}\p_{\alpha}$  changes nontrivially under (\ref{52}) and we do not know which new affine distance $r'$ is constant on $H$. Still, a freedom of function $\alpha$  may be sufficient to obtain condition (\ref{34a}) in some  coordinates. Then the Penrose inequality must be satisfied for any system of  coordinates $u,r,x^A$. Moreover, surfaces $u=const$ in coordinates satisfying (\ref{34a}) could be used to define a system of ``good cuts'' of the scri $J^+$. In this way the BMS group would be broken to the Lorentz group (or its subgroup) what is important for a definition of the angular momentum (see \cite{sz} for a review).

\section{Energy estimates in the Kerr metric}
In the Boyer-Lindquist coordinates the Kerr solution with $a<m$ reads 
\be 
g=(1-\frac{2mr}{\rho^2})dt^2+\frac{4mar\sin^2{\theta}}{\rho^{2}}dtd\varphi-\frac{\Sigma^2\sin^2{\theta}}{\rho^2}d\varphi^2-\frac{\rho^2}{\Delta}dr^2-\rho^2d\theta^2\ ,
\label{54}
\ee
where
\be
\rho^2=r^2+a^2\cos^2{\theta}\ ,\ \ \Delta=r^2-2mr+a^2\ ,\ \Sigma^2=(r^2+a^2)^2-a^2\Delta\sin^2{\theta}\ .\label{55}
\ee
The event horizon is located at $r=r_H$, where 
\be\n{55a}
r_H=m+\sqrt{m^2-a^2}\ .
\ee

The Kerr metric takes the form (\ref{1a}) in (singular) coordinates found by Fletcher and Lun \cite{fl}. Below we derive an equivalent system of coordinates  in a simplified way. We start with the ansatz
\be
u=t+f_1(r)+f_2(\theta)\ .\n{56}
\ee
Then from equation $u_{,\alpha}u^{,\alpha}=0$ one obtains
\be
u=t+\epsilon(\int{\frac{\tilde \Delta}{\Delta}dr}+\tilde\epsilon a \sin{\theta})\ ,\ \ \tilde\epsilon=\pm 1\ ,\n{57}
\ee
where 
\be
\tilde \Delta=\sqrt{r^4+a^2r^2+2a^2mr}\n{59}
\ee
and $\epsilon=-1$ for the retarded time and $\epsilon=1$ for the advanced time. 
 Coordinate $u$ becomes the  Eddington-Finkelstein time  in the Schwarzschild metric if $a=0$. For $a\neq 0$, surfaces given by $u=const$ and $t=const$ have conical singularities at $\theta=0,\pi$ directed to the exterior ($\tilde\epsilon=1$) or to the interior ($\tilde\epsilon=-1$) of these surfaces. They smooth out when we approach the horizon. A freedom of $\tilde\epsilon$ seems to be a new element with respect to coordinates defined in \cite{fl}.

 Coordinates $\tilde\theta$ and $\tilde \varphi$ should be anihilated by  vector  $u^{,\alpha}\p_{\alpha}$. Functions of this type can be found again by  separation of  variables. A preferable choice for $\tilde \varphi$ is 
 \be
 \tilde \varphi=\varphi-2\epsilon am\int_r^{\infty}{\frac{r'}{\Delta\tilde\Delta}dr'}\n{60}
 \ee
 and for $\tilde \theta$ one can take a function of the variable
 \be\n{62}
 \xi=e^{\tilde\epsilon a f}\tan{(\frac{\theta}{2}-\frac{\pi}{4})}\ ,
 \ee
 where
 \be\label{62a}
 f(r)=\int_{r_H}^r{\frac{dr'}{\tilde \Delta}}\ .
 \ee
 For our purposes it is convenient to define
 \be
 \tilde \theta=2\arctan{\xi}+\frac{\pi}{2}\n{61}
 \ee
 (this coordinate differs slightly from that in \cite{fl}).  Note that $\xi$  takes values from the interval $[-e^{\tilde\epsilon a f},e^{\tilde\epsilon a f}]$ which depends on $r$.

   In order to find an affine coordinate $\tilde r$, we have to solve equation $u^{,\alpha}\p_{\alpha}\tilde r=1$. Its particular solution of the form $f_3(r)+f_4(\theta)$ is given by
  \be\n{64a}
 \int_{r_H}^r{\frac{r'^2}{\tilde \Delta}dr'}+\tilde\epsilon a\sin{\theta}\ .
 \ee
 It can be composed with a function of $\tilde\theta$ and $\tilde\varphi$. In this way 
it is easy to construct an affine distance $\tilde r$ which equals $r_H$ on $H$ 
  \be\n{64}
 \tilde r=\int_{r_H}^r{\frac{r'^2}{\tilde \Delta}dr'}+\tilde\epsilon a(\sin{\theta}-\sin{\tilde\theta})+r_H\ .
 \ee
 Coordinates $\tilde r$ and $\tilde\theta$ coincide, respectively,  with $r$ and $\theta$ on the horizon. For $a>0$, coordinate  $\tilde \varphi$ is logarithmically divergent near $H$ and it tends to $\varphi$ if $r\rightarrow\infty$. Derivatives of $\tilde u$ and $\tilde r$ with respect to $\theta$ are finite but nonvanishing at $\theta=0,\pi$. Hence, differentials of these coordinates are divergent on the symmetry axis. For $a=0$  coordinates ($\tilde r$, $\tilde\theta$, $\tilde\varphi$) coincide with the standard coordinates ($ r$, $\theta$, $\varphi$) of the Schwarzschild metric.
 
In  the  coordinates $u$, $\tilde r$, $\tilde\theta$, $\tilde \varphi$ the Kerr metric becomes more complicated than in the Boyer-Lindquist coordinates. In order to get a compact expression we  write it down using components $g_{tt},g_{t\varphi},g_{\varphi\varphi}$ of  (\ref{54})
\ba\label{65}
g=du\big[g_{tt}du-2\epsilon d\tilde r+2g_{t\varphi}d\tilde\varphi-2\epsilon\tilde\epsilon a(\cos^2{\tilde\theta}-\frac{2mr}{\rho^2}\cos^2{\theta})d\tau\ \big]\\\nonumber
+g_{\varphi\varphi}d\tilde\varphi^2-(r^2+\frac{2ma^2r}{\rho^2}\cos^2{\theta})\cos^2{\theta}d\tau^2- \frac{4\epsilon\tilde\epsilon ma^2r}{\rho^{2}}\sin^2{\theta}\cos^2{\theta}d\tilde \varphi d\tau\ ,
\ea
where
\be
d\tau=\frac{d\tilde\theta}{\cos{\tilde\theta}}\ .
\ee
Here, $r$ and $\theta$ should be considered as functions of $\tilde r$ and $\tilde \theta$ defined by relations (\ref{61}) and (\ref{64}). For $\epsilon=-1$, metric (\ref{65}) takes  form (\ref{1}). It has no singularity at $r=r_H$, so it can be continued through the past event horizon. In order to continue metric through the future horizon, we should use coordinates  corresponding to $\epsilon=1$.

Although coordinates $\tilde x^{\mu}$ are singular, there are strong indications that inequality (\ref{34}) should be satisfied. Conditions (i)-(iii) of section 4 are satisfied with $\beta=\Delta$. One can show that equations (\ref{13}) and (\ref{14}) are true, if we include in $R^{(2)}$ the Dirac delta terms at $\theta=0,\pi$. It follows from (\ref{65}) that functions $\omega^A$ are finite at $\theta=0,\pi$. Hence, inequality (\ref{19a}) is still true and one can generalize  the first part of Proposition 4.2 to the present case. We cannot apply its second part, since condition (\ref{34a}) is not fulfilled.  Instead,   we will  compute all ingredients of (\ref{34}) numerically  and compare the corresponding estimates  for $\tilde\epsilon=1$ and $\tilde\epsilon=-1$  with the Penrose inequality (\ref{35}).

Coordinate $\tilde\theta$ is convenient to describe the Kerr metric near the horizon but, in order to analyze metric when $u=const$ and $\tilde r\rightarrow\infty$, it is more convenient to use coordinate $\theta'$ given by
\be\n{66}
 \tan{(\frac{\tilde\theta}{2}-\frac{\pi}{4})}=e^{\gamma}\tan{(\frac{\theta'}{2}-\frac{\pi}{4})}\ ,
 \ee
 where
 \be
 \gamma=\tilde\epsilon a \int_{r_H}^{\infty}{\frac{dr}{\tilde \Delta}}\ .\n{67}
 \ee
 Function $n$ defined by (\ref{3a}) is invariant  under this change.
 Coordinate $\theta'$ tends to $\theta$ at infinity and it coincides with the  coordinate $\theta$ in \cite{fl}. Function $\sigma$ (see (\ref{3b})) can be obtained either directly from (\ref{65}) or, more easily, by  transforming  the determinant of the Kerr metric (\ref{54}) to the coordinates $u$, $\tilde r$, $\theta'$ and $\tilde \varphi$. It reads
 \be\n{69}
 \sigma=\tilde\Delta\sin{\theta}\frac{\cos{\theta}}{\cos{\theta'}}\ .
\ee 
 Expressing $\sigma$ in the new coordinates  and expanding it according to   (\ref{53a})  yields $\sigma_s=\sin{\theta'}$ and
  \be\n{68}
 n=4\int_{r_H}^{\infty}{(1-\frac{r^2}{\tilde \Delta})dr}+4\tilde\epsilon a\frac{\cos^2{\theta'}}{\sin{\theta'}-\coth{\gamma}}-2\tilde\epsilon a\frac{\cos{2\theta'}}{\sin{\theta'}}\ .
 \ee
  As expected, formula (\ref{6a}) implies $M_B=m$, hence $E(u)=m$. 
 Function $n$ is singular at $\theta'=0,\pi$, however it is integrable over the sphere. We calculated its mean value $\langle n\rangle$ using Mathematica 10. For both values  $\tilde\epsilon=\pm 1$, inequality (\ref{34}) is satisfied and $\langle n\rangle\geq 0$. Thus, inequality (\ref{34}) is stronger than the Penrose inequality (\ref{35}) which in the case of the Kerr metric is equivalent to $m\geq \frac 12r_H$. We obtain the following hierarchy of estimates
 \be\n{70}
 m\geq\frac 12r_H+\frac 18\langle n \rangle_{-1}\geq\frac 12r_H+\frac 18\langle n \rangle_{1}\geq\frac 12r_H\ ,
 \ee
 where subscripts $\pm 1$ correspond to $\tilde\epsilon=\pm 1$. All inequalities are saturated for $a=0$. Differences between estimates grow monotonically when $a/m\rightarrow 1$. For $a=m$ relation (\ref{70}) reads (approximately)
 \be\n{71}
 m>0,74m>0,63m>0,50m\ .
 \ee

 In the  case of the Kerr metric, we are not able to find  coordinates which satisfy condition  (\ref{34a}). A difference between expression (\ref{68})  for $\tilde\epsilon=-1$ and $\tilde\epsilon=1$ shows that an action of the BMS supertranslations (\ref{52}) is nontrivial and achieving condition (\ref{34a}) is not excluded. On the other hand,  inequality (\ref{70}) shows that this condition  may be not necessary for a  derivation of  the Penrose inequality from (\ref{34}).

\section{Summary}
We have been considering the Trautman-Bondi energy $E(u)$ of  null surfaces $u=const$ within the variant of the Bondi-Sachs formalism with the luminosity distance  replaced by the affine distance $r$. This energy  can be written as the limit (\ref{11a}) of a mean  value of null expansions of surfaces  of constant $u$ and $r$. If the dominant energy condition is satisfied, then the energy inequality (\ref{19a})  follows from one of the Einstein equations. It simplifies to (\ref{46b}) for a special affine coordinate $\hat r$. In section 4 we assume  existence of a regular event horizon $H$ such that appropriately scaled  null vector field tangent to $u=const$ extends to $H$. Then there exists   coordinate $r$ which is constant on $H$ (Proposition 4.1), exactly as in  the case of the Schwarzschild solution. Inequality (\ref{34}) follows in the limit of  (\ref{19a}) when the horizon is approached (Proposition 4.2).  Under assumption (\ref{34a}) about a long time 
behavior of the asymptotic 
data, inequality (\ref{34}) converts into the Penrose inequality (\ref{35}). In section 5, we present coordinates for the Kerr metric closely related to  the generalized Bondi-Sachs coordinates of Fletcher and Lun. They have divergent derivatives at $\theta=0,\pi$. Nevertheless, they can be used to illustrate the results of the preceding sections. Energy  inequality (\ref{34}) is satisfied and it is stronger than the Penrose inequality.

\end{document}